\documentclass[prd,aps,amsfonts,nofootinbib,notitlepage,longbibliography, twocolumn]{revtex4-1}

\usepackage{graphicx}
\usepackage{tabularx}
\usepackage{braket}
\usepackage{amssymb}
\usepackage{amsmath}
\usepackage{bm}
\usepackage[driverfallback=dvipdfm]{hyperref}
\hypersetup{colorlinks=true,breaklinks,urlcolor=blue,linkcolor=blue,citecolor=blue}

\newcommand{\comment}[1]{}

\newcommand{\lr}[1]{\left( #1\right)}
\newcommand{\mlr}[1]{\left[ #1\right]}

\newcommand{\norm}[1]{\left\lVert#1\right\rVert}
\newcommand{\abs}[1]{\left\lvert#1\right\rvert}

\newcommand{\ii}{\mathrm{i}}
\newcommand{\ee}{\mathrm{e}}

\newcommand{\tr}[1]{\mathrm{Tr}\lr{#1}}
\newcommand{\Tr}{\mathrm{Tr}}

\newcommand{\order}{\mathrm{O}}

\newcommand{\cZ}{\mathcal{Z}}

\newcommand{\W}{{\mathbf{W}}}
\newcommand{\V}{{\mathbf{V}}}
\newcommand{\poly}[1]{\mathrm{poly}\lr{#1} }
\newcommand{\fkd}{\mathfrak{d}}

\makeatletter
\renewcommand{\p@subsection}{}
\renewcommand{\p@subsubsection}{}
\makeatother

\usepackage{amsthm}
\newtheorem{thm}{Theorem}
\newtheorem{cor}[thm]{Corollary}
\newtheorem{lem}[thm]{Lemma}

\newtheorem{defn}[thm]{Definition}

\begin{document} 

\title{Polynomial-time classical sampling of high-temperature quantum Gibbs states}
\author{Chao Yin}\email{chao.yin@colorado.edu}
\affiliation{Department of Physics and Center for Theory of Quantum Matter, University of Colorado, Boulder CO 80309, USA}

\author{Andrew Lucas}\email{andrew.j.lucas@colorado.edu}
\affiliation{Department of Physics and Center for Theory of Quantum Matter, University of Colorado, Boulder CO 80309, USA}

\date{\today}

\begin{abstract}
    The computational complexity of simulating quantum many-body systems generally scales exponentially with the number of particles.  This enormous computational cost prohibits first principles simulations of many important problems throughout science, ranging from simulating quantum chemistry to discovering the thermodynamic phase diagram of quantum materials or high-density neutron stars.   We present a classical algorithm that samples from a high-temperature quantum Gibbs state in a computational (product state) basis.  The runtime grows polynomially with the number of particles, while error vanishes polynomially.  This algorithm provides an alternative strategy to existing quantum Monte Carlo methods for overcoming the sign problem. 
    Our result implies that measurement-based quantum computation on a Gibbs state can provide exponential speed up only at sufficiently low temperature, and further constrains what tasks can be exponentially faster on quantum computers. 
\end{abstract}

\maketitle

\section{Introduction}
Many-body physics is greatly aided by the ability to use classical computation to perform calculations that cannot be done by hand.  When one wishes to calculate the properties of a system at finite temperature $T = 1/\beta$, it is desirable to have an algorithm that efficiently samples from a thermal (Gibbs) distribution, where  the probability $p(s)$ of finding state $s$ is given by \begin{equation} \label{eq:ps}
    p(s) = \frac{\langle s| \mathrm{e}^{-\beta H} |s\rangle}{\tr{\mathrm{e}^{-\beta H}}}.
\end{equation}
For simplicity in this paper, we consider states $|s\rangle$ that are bit strings of length $N$: i.e., the quantum system consists of $N$ interacting qubits.

Even in classical statistical mechanics, this \emph{Gibbs sampling} task is generically hard.  We can easily compute the numerator of (\ref{eq:ps}): $H(s)$ is a known \emph{function}.   But if ground states of $H$ encode solutions to an NP-hard optimization problem \cite{Barahona,Lucas_2014}, estimating the denominator of (\ref{eq:ps}) may take exponential runtime.  Hence, efficient sampling at low temperature is restricted to special ferromagnetic models without frustration \cite{samp2count86,ising_poly93,ising_samp99}.  In contrast, at high temperature, so long as $H$ contains only few-body interactions (with each spin not coupled to extensively many others), we can accurately sample from $p(s)$ at sufficiently small $\beta$ using standard Glauber/Metropolis dynamics \cite{levin}, using polynomial runtime. 

It has not been established whether or not quantum systems with few-body ($k$-local) interactions are easy to sample classically, even at high temperature.  We can immediately see why quantum mechanics makes this task harder: since $H$ may be a sum of local \emph{non-commuting} terms, evaluating both the numerator and the denominator of (\ref{eq:ps}) become hard. This remains true even if all that is desired is the thermal expectation value of a local observable: \begin{equation}\label{eq:thermalcorr}
        \langle A\rangle_\beta := \frac{\tr{A \mathrm{e}^{-\beta H}}}{\tr{\mathrm{e}^{-\beta H}}}.
    \end{equation}
A common setting where this expectation value is notoriously hard is in systems with interacting bosons and fermions.  Integrating out fermions, one finds formally exact expressions for both the numerator and denominator of (\ref{eq:thermalcorr}).  However, these expressions cannot be evaluated easily, as they are over rapidly oscillating functions causing the \emph{sign problem} \cite{troyer}.  Nevertheless, this numerical method (called \emph{quantum Monte Carlo}) is often the state-of-the-art, even though its runtime could be exponential in $N$ for any $\beta>0$ \cite{dornheim}.  The generic challenge of quantum Monte Carlo makes it hard to study a broad array problems, such as simulating quantum chemical reactions \cite{anderson}, or determining the phase diagrams of high-density nuclear matter \cite{creutz,carlson} or  quantum materials \cite{hubbard}. 


So far, rigorous results exist for quantum systems remain limited.   There are quantum systems, e.g. systems without a sign problem, where one can prove that quantum Monte Carlo -- including classical sampling of $p(s)$ --  is efficient \cite{bravyi2007complexity,MBQC_thermal09,QMC_ferro17,Li_2019,QMC_Ising20,QMC_1d21}.  A first \emph{generic} rigorous locality result for high-temperature Gibbs states was \cite{local_T14}, which proves that connected correlators decay exponentially. As an improvement,
\cite{cond_mutu_info20} argues that high-temperature Gibbs states are approximate Markov networks. Combining with \cite{Brandao_prepare19}, this yields a $\log N$ depth \emph{quantum} algorithm to prepare the state. However, it does not guarantee the state can be sampled classically in polynomial time, since even depth-3 circuits are known to be exponentially hard to simulate on a classical computer \cite{depth3_hard04} in the worst-case.  \cite{Harrow_complex0_20} also studied the problem of classically sampling a quantum Gibbs state, but did not obtain a polynomial-time algorithm.
Further efforts to understand the computational complexity of characterizing quantum Gibbs states can be found in \cite{bravyi_denseH22,thermal_tensor_local21,thermal_cor1d22,Chen_prepare23}, and the review \cite{thermal_rev22}.  

  The purpose of this work is to show that, for arbitrary quantum systems with geometrically local few-body interactions, classical sampling of quantum systems can be done in polynomial time in $N$, at sufficiently small $\beta = \mathrm{O}(1)$ (i.e. $\beta$ does not need to vanish at large $N$).  Since, as explained above, at sufficiently large $\beta = \mathrm{O}(1)$, quantum systems must be hard to sample in general, our work establishes a relatively abrupt transition from easy to hard to simulate. Our result provides strong constraints on what tasks a future quantum computer may perform more efficiently than a classical computer.




\section{Main results}

We consider a set of $N$ qubits (two-level systems), $\Lambda=\{1,\cdots,N\}$, placed on the vertices of a graph $G$, whose maximal degree $K=\mathrm{O}(1)$ (does not scale with $N$). This includes constant-dimensional lattices, used to simulate lattice gauge theories or correlated electrons, as special cases. We define the distance $\mathsf{d}(i,j)$ between qubits $i$ and $j$ as the minimal number of edges contained in a path connecting $i$ and $j$. This induces the diameter of any set of qubits $S\subset \Lambda$ to be defined as ${\rm diam}(S) := \max_{i,j\in S}\mathsf{d}(i,j)$. 

We assume the qubits interact via local Hamiltonian \begin{equation}\label{eq:H=Ha}
    H = \textstyle{\sum_{a}} \lambda_a H_a,
\end{equation}
where each term $H_a$ is a Pauli string acting on a set $S_a\subset \Lambda$ of diameter no larger than some constant, with $\abs{S_a}\le k$.  In computer science, such Hamiltonians are called $k$-local. We further assume that each term is bounded ($\abs{\lambda_a}\le 1$), and that single-site Pauli matrices $\{X_1,Y_1,Z_1,X_2,\cdots\}$ are all included in $\{H_a\}$. As a result, there exists a constant $\fkd$ such that, each term $H_a$ is supported on a local region that overlaps with the support of at most $\fkd$ other terms, with $\fkd < 4^k \Delta! / k!(\Delta-k)!$ where $\Delta$ is the maximal degree of a vertex (the number of qubits $v$ for which -- for fixed $u$ -- we can find $\lbrace u,v\rbrace \subseteq S_a$ for some set $S_a$). By considering a single-site Pauli as one of the allowed $H_a$ terms, we see that a (potentially loose) upper bound on the number of $a$s in \eqref{eq:H=Ha} is $\fkd N$. While the discussion of this paper focuses on qubits, our results can be generalized to other finite-dimensional quantum systems. 

 Let us assume for now (we will generalize later) that we wish to sample from a probability distribution obtained by the following thought experiment.  We draw a state $|\psi\rangle$ from the quantum Gibbs ensemble $\mathrm{e}^{-\beta H}$ (up to normalization), and then simultaneously measure the commuting Paulis $Z_1,\ldots, Z_N$. This collapses the wave function into the computational basis state $|x_1\cdots x_N \rangle := |\bm{x}\rangle $.  We use the convention that (for one qubit) $Z|0\rangle = |0\rangle$, $Z|1\rangle = -1$.  Denoting $E_{0_i} = \frac{1+Z_i}{2}$ and $E_{1_i} = \frac{1-Z_i}{2}$, our goal is  to sample the binary string $\bm{x}$ from the probability distribution \begin{equation}\label{eq:px=}
    p(\bm{x}) := \frac{\tr{E_{x_1}\cdots E_{x_N} \mathrm{e}^{-\beta H}}}{\tr{\mathrm{e}^{-\beta H}}}.
\end{equation}
Our main result is that for sufficiently small $\beta$ (high temperature), this task is ``easy" on a classical computer:

\begin{thm}\label{thm:main}
    Suppose \begin{equation}\label{eq:beta*}
        \beta < \beta_* := \mlr{ 2\ee^2 \fkd(\fkd+1) }^{-1}.
    \end{equation}
    There exists a classical algorithm with runtime $t\sim \mathrm{poly}(N)$, that samples the measurement outcome $\bm{x}$ from probability distribution $p'(\bm{x})$ that is close to the true $p(\bm{x})$ in total variation distance \begin{equation}\label{eq:p'-p<poly}
        \textstyle{\sum_{\bm{x}} } \abs{p'(\bm{x})-p(\bm{x})} =: \norm{p'-p}_1 \le 1/\poly{N}.
    \end{equation}
\end{thm}

Let observable $A$ be supported on a set $R$ containing finitely many vertices.  By just explicitly evaluating \begin{equation}
\langle A\rangle \approx \sum_{\bm{x}} p'(\bm{x}) \langle \bm{x}|A|\bm{x}\rangle,  
\end{equation}
Theorem \ref{thm:main} implies:
\begin{cor}\label{cor:local}
    For arbitrary local operator $A$, there exists a classical algorithm with runtime $t\sim \mathrm{poly}(N)$ to calculate thermal expectation values (\ref{eq:thermalcorr})
    up to an inverse-polynomial error.
\end{cor}
Corollary \ref{cor:local} was previously argued for in \cite{Brandao_prepare19}, although their methods did not establish our stronger result Theorem \ref{thm:main}.  Indeed, our Theorem \ref{thm:main} implies that \emph{generic quantum Monte Carlo calculations} can be done in polynomial time, at high enough temperature.  

In a related fashion, while it was shown that thermal states at sufficiently low temperature can be resources for universal measurement-based quantum computation \cite{MBQC_thermal09,MBQC_thermal11,MBQC_thermal13}, the following theorem forbids this possibility at high temperature, by extending \cite{local_T14} to show that \emph{even after measuring arbitrarily many qubits}, there are no long-range correlations in the resulting quantum state.

\begin{thm}\label{thm:nocorr}
At high-temperature \eqref{eq:beta*}, after projecting $n$ qubits onto fixed computational basis states using $E_n$, two unmeasured qubits $i,j$ have correlation exponentially small in their distance
\begin{align}\label{eq:cor<}
    &\mathrm{Cor}(i,j):=\max_{E_i,E_j} \frac{ \tr{E_iE_j E_n \ee^{-\beta H}} }{ \tr{E_n \ee^{-\beta H}}}- \nonumber\\ & \frac{ \tr{E_iE_n \ee^{-\beta H}} }{ \tr{E_n \ee^{-\beta H}}}
    \frac{ \tr{E_j E_n \ee^{-\beta H}} }{ \tr{E_n \ee^{-\beta H}}} \le c'' \lr{\frac{\beta}{\beta_*}}^{c'\mathsf{d}(i,j)},
\end{align}
for some constants $c',c''>0$,
where $E_i,E_j$ act on $i$ and $j$ respectively with $\norm{E_i},\norm{E_j}\le 1$.
\end{thm}

\section{Algorithm}
\label{sec:samp2comp}
We will prove Theorem \ref{thm:main} constructively: i.e., provide the explicit algorithm that gives efficient sampling. The algorithm is inspired by the following idea: suppose that we could easily compute the marginal probabilities $p(x_j|\bm{x}_{j-1})$, where $\bm{x}_{j-1}=(x_1,\ldots, x_{j-1})$.  We have assumed an (arbitrary) ordering of the qubits.  Since \begin{equation}
    p(\bm{x}) = p(x_1)p(x_2|x_1)\cdots p(x_N|\bm{x}_{N-1}),
\end{equation}
we could then (in $\mathrm{O}(N)$ time) calculate the actual distribution $p(\bm{x})$.  Of course, what we would really like is just to \emph{sample} from the distribution $p(\bm{x})$, but that is also readily done.  We simply pick $x_1$ with probability $p(x_1)$, and conditioned on our result, pick $x_2$ with probability $p(x_2|x_1)$, etc.  This algorithm also has $\mathrm{O}(N)$ runtime, and -- by construction -- samples \emph{exactly} from the desired distribution $p(\bm{x})$.

Our key observation is that while we do not know how to exactly calculate $p(x_j|\bm{x}_{j-1})$ in general, we can calculate it \emph{approximately}.  Our first lemma shows that we only need to calculate these marginals with polynomial accuracy, to sample from $p(\bm{x})$ with accuracy (\ref{eq:p'-p<poly}).


\begin{lem}
    Suppose for each substring $\bm{x}_n$, a number $0 \le p'(0| \bm{x}_n) \le 1$ is known such that  \begin{equation}\label{eq:p'-p<e}
    \abs{p'(0| \bm{x}_n) - p(0| \bm{x}_n)}\le \epsilon\sim N^{-\alpha}.
\end{equation}
Then, if $\alpha>1$, the distribution \begin{equation}
  p'(\bm{x}) = p'(x_1)p'(x_2|x_1)\cdots p'(x_N | \bm{x}_{N-1})
\end{equation}
satisfies \eqref{eq:p'-p<poly}.
\end{lem}
\begin{proof}
We follow \cite{samp2comp17}.
    Let $p'(1| \bm{x}_n) = 1- p'(0| \bm{x}_n)$. 
    To bound the total variation distance $\lVert p-p'\rVert_1$, we re-write \begin{align}\label{eq:p-p=tele}
    &p'(\bm{x})-p(\bm{x}) = \nonumber\\ & p(x_1)\cdots p(x_{N-1}| \bm{x}_{N-2})
    \mlr{p'(x_N| \bm{x}_{N-1}) - p(x_N| \bm{x}_{N-1})} \nonumber\\ &+ \cdots 
    + \mlr{p'(x_1)-p(x_1)}p'(x_2|x_1)\cdots p'(x_N| \bm{x}_{N-1}).
\end{align}
We now observe that:
\begin{align}\label{eq:sump<2e}
  \textstyle{\sum_{\boldsymbol{x}}}&| p'(\bm{x}) - p(\bm{x})| \le   \textstyle{\sum_{\bm{x}}} \sum_j p(x_1)\cdots p(x_{j-1}|\bm{x}_{j-2}) \notag \\
  & \cdot | p'(x_j|\bm{x}_{j-1}) - p(x_j|\bm{x}_{j-1})| p'(x_{j+1}|\bm{x}_j) \cdots p'(x_N|\bm{x}_{N-1}) \notag \\
  &\le \epsilon \textstyle{\sum_j \sum_{x_1\cdots x_{j}}}p(x_1)\cdots p(x_{j-1}|\bm{x}_{j-2}) \le 2\epsilon N.
\end{align} 
The first inequality is simply the triangle inequality; the second uses that the iterative sum over marginals for $x_N,\ldots, x_{j+1}$ gives exactly 1, and invokes (\ref{eq:p'-p<e}); the final inequality follows analogously to the second, together with the fact that the sum over $j$ contains $N$ terms. 
\end{proof} 

\section{Calculating marginal probabilities}
It remains to show that, using polynomial classical resources and runtime, we can calculate an approximation of the marginal $p(0|\bm{x}_n)$ with error $N^{-\alpha}$ for any $\bm{x}_n$. We emphasize that we do not need to compute the marginal for all substrings, just the one given the string we have found so far in any given run of the algorithm.

Similar to \eqref{eq:px=}, we have  \begin{equation}\label{eq:marginal0}
    p(0|\bm{x}_n) = \frac{ \tr{E_{x_1}\cdots E_{x_n} E_{0_j} \ee^{-\beta H}} }{ \tr{E_{x_1}\cdots E_{x_n} \ee^{-\beta H}}},
\end{equation}
where we set $j=n+1$. In a nutshell, we want to calculate the local density matrix on qubit $j$, after measuring the previous $n$ qubits with outcome $\bm{x}_n$.  Recent work using cluster expansions \cite{learn_highT21,Loschmidt_echo23} has shown that such few-qubit reduced density matrices can be accurately estimated in the absence of the extensive product $E_{x_1}\cdots E_{x_n}$.  The basic strategy is to notice that \begin{align}
    2&p(0|\bm{x}_n)-1 =\notag \\
    &\left.\frac{\partial}{\partial \kappa} \log \left(\tr{E_{x_1}\cdots E_{x_n}\mathrm{e}^{\kappa Z_j} \mathrm{e}^{-\beta H }}\right)\right|_{\kappa=0}.
\end{align}
The logarithm of a partition function, or free energy, can be efficiently estimated by cluster expansion at high $\beta$: the terms which contribute to this sum correspond to traces over products of $H_a$s that form \emph{connected clusters}.  Moreover, since we are interested only in the $\kappa$ derivative of this sum, we can focus on connected clusters that include site $j$. Notice that these facts do not depend on the projectors $E_{x_1}\cdots E_{x_n}$.  Therefore, we can extend Theorem 3.1 of \cite{learn_highT21} to derive:
\begin{lem}\label{lem:clus}
    For $\beta$ obeying \eqref{eq:beta*}, the marginal is given by
    \begin{equation}\label{eq:p=beta}
        p(0|\bm{x}_n) = \frac{1}{2}+ \frac{1}{2}\sum_{m=1}^\infty \gamma_m \beta^m,
    \end{equation}
    where $\gamma_m$ is bounded: \begin{equation}\label{eq:gamma<beta}
        \abs{\gamma_m} \le m \beta_*^{-m},
    \end{equation}
    with $\beta_*$ given in \eqref{eq:beta*}. This implies the series \eqref{eq:p=beta} converges absolutely. Moreover, $\gamma_m$ can be computed in time 
    $\order\mlr{\exp(cm)}$ where $c=\log\lr{ \fkd 2^{2k+1}}$ is a constant; recall the definitions of $\fkd$ and $k$ below \eqref{eq:H=Ha}.
\end{lem}

The proof of Lemma \ref{lem:clus} is technical (but constructive), and we explain it in Appendix \ref{app:A}.  With this lemma in hand, we can explain how to quickly sample from $p'(\bm{x})$.  One approximates $p(0|\bm{x}_n)$ for $n=0,\ldots,N-1$ by simply truncating the sum in \eqref{eq:p=beta} to $m\le M= \eta\log N$. Using the bound \eqref{eq:gamma<beta}, $\eta$ is related to $\alpha$ defined in \eqref{eq:p'-p<e} by \begin{equation}
    \lr{\log N} (\beta/\beta_*)^M \lesssim N^{-\alpha}, 
\end{equation}
meaning that \begin{equation}
    \eta > \alpha \lr{\log \frac{\beta_*}{\beta}}^{-1}.
\end{equation}
According to Lemma \ref{lem:clus}, this takes computation time $\order\mlr{\exp(cM)} = \order(N^{c\eta})$.
Combining this computation algorithm with the reduction in Section \ref{sec:samp2comp} yields the advertised sampling algorithm in Theorem \ref{thm:main}, with runtime $\order(N^{1+c\eta})$ to output one string $\bm{x}$.

\section{Generalization to adaptive protocols}
Previously we considered a thought-experiment where we measured the Gibbs state in the computational basis; the result was easy to sample.  In quantum mechanics, we might measure individual qubits in arbitrary bases, and change the basis of future measurement based on previous outcomes. Such adaptive protocols play a critical role in error correction and measurement-based quantum computation. Nevertheless, at sufficiently high temperature (\ref{eq:beta*}), we now show that classically simulating this quantum adaptive protocol is easy, at high temperature.

Suppose one first measures qubit $n=1$ in some basis, and then depending on the measurement outcome $x_1\in \{0,1\}$, qubit $n=2$ is measured in an adaptive basis to get $x_2$. Then one measures $n=3$ in a basis determined by the previous outcomes $x_1,x_2$, and so on. This defines adaptive local projectors $E_{x_1\cdots x_n}$ acting on qubit $n$, and the outcome binary string $\bm{x}$ is sampled from probability $p(\bm{x})$ defined by marginals \begin{equation}\label{eq:marginal}
    p(x_{n+1}|\bm{x}_n) = \frac{ \tr{E_{\bm{x}_{n+1} } \ee^{-\beta H}} }{ \tr{E_{\bm{x}_n } \ee^{-\beta H}}},
\end{equation}
where $E_{\bm{x}_n }:= E_{x_1}E_{x_1x_2}\cdots E_{x_1\cdots x_n}$. Note that which qubit to measure next can depends on previous measurement outcomes, where $n$ would become a dynamical label.

\comment{
This measurement procedure is described as follows. First, a projective measurement is done on a chosen qubit that we label by $n=1$, such that depending on the measurement outcome $x_1\in \{0,1\}$, it is projected to the local state $\ket{e_{x_1}}_1$ ($\{\ket{e_0}_1,\ket{e_1}_1\}$ is a local basis). Mathematically, denote the measurement projector $E_{x_1}:= \ket{e_{x_1}}_1 \bra{e_{x_1}}$, the global state becomes $E_{x_1}\ee^{-\beta H} E_{x_1} $ up to normalization. The probability to get outcome $x_1$ is \begin{equation}
    p(x_1) = \frac{ \tr{E_{x_1}\ee^{-\beta H}} }{ \tr{\ee^{-\beta H}}}.
\end{equation}
Second, measure qubit $n=2$ in a local basis $\{\ket{e_{x_1,0}}_2,\ket{e_{x_1,1}}_2\}$ that may depend on the previous outcome $x_1$. This second qubit need not be close to qubit $n=1$ in the graph. Let $E_{x_1x_2}:= \ket{e_{x_1x_2}}_2\bra{e_{x_1x_2}}$, this measurement changes the state to $E_{x_1x_2} E_{x_1}\ee^{-\beta H} E_{x_1} E_{x_1x_2}$ with conditional probability (given $x_1$) \begin{equation}
    p(x_2|x_1) = \frac{p(x_1x_2)}{p(x_1)} = \frac{ \tr{E_{x_1x_2} E_{x_1}\ee^{-\beta H}} }{ \tr{E_{x_1} \ee^{-\beta H}}}.
\end{equation}
Then given $x_1,x_2$, one determines the local basis to measure $n=3$, and so on.
}

Assuming the adaptive local basis can be computed efficiently from previous measurement outcomes, we claim Theorem \ref{thm:main} still holds in this adaptive case. The reason is that one can still use the algorithm in Section \ref{sec:samp2comp}, where marginals are computable using Lemma \ref{lem:clus}, which we prove for the general adaptive case in Appendix \ref{app:A}. 
The proof of Theorem \ref{thm:nocorr} uses similar methods, and is found in Appendix \ref{app:B}.

\section{Outlook}
We provide an efficient (polynomial runtime in $N$, the number of qubits) algorithm for sampling from quantum Gibbs states \emph{on a classical computer}.  There exists a small (but still O(1)) $\beta$ such that our algorithm runtime is $t\sim N^{1+\delta}$ for any constant $\delta>0$.  Our proof is inspired by recent technical advances \cite{learn_highT21,Loschmidt_echo23} in applying combinatorial cluster expansion techniques to many-body quantum systems.  Our work provides a rigorous foundation for the physically intuitive idea that, at very high temperature, quantum systems are in a ``disordered phase" that is easy to classically simulate, together with explicit algorithms for such simulation. 

Our algorithm may provide an interesting alternative and/or subroutine for quantum Monte Carlo methods.  After all, while Theorem \ref{thm:main} gives a bound on when our algorithm is \emph{guaranteed} to perform well, it might be the case that in certain models, our sampling algorithm remains efficient down to lower temperatures.  We hope that these ideas can find practical applications in large-scale computation in the coming years.

Using generalizations of the Jordan-Wigner transformation \cite{Bravyi_2002,cirac,setia,derby}, we can encode local Hamiltonians involving fermion operators using local qubit Hamiltonians on an enlarged Hilbert space.  However, such mappings encode the fermionic states into an entangled subspace of Hilbert space, meaning that our sampling algorithm in a product basis of the qubit Hilbert space will not immediately apply.  It would be interesting to generalize our methods to such fermionic models.

We will report in upcoming work \cite{us_metrology} that it is classically efficient to sample from the probability distribution $|\langle \bm{x}|\psi\rangle|^2$, where $\ket{\psi}=\ee^{-\ii t H} \ket{\bm{0}}$ is a time-evolved product state for a short time $t$ and $|\bm{x}\rangle$ is sampled from certain bases, along with a generalization to settings where $H(t)$ is time-dependent.  Since it is known that the simulation becomes hard in the worst case if $t$ is larger than some constant (e.g. if $\ket{\psi}$ is a cluster state that enables measurement-based quantum computation \cite{oneway_qc01}), this also demonstrates a sharp transition in the simulatability of quantum evolution in real time.

\section*{Acknowledgements}
This work was supported by the Alfred P. Sloan Foundation under Grant FG-2020-13795 (AL) and by the U.S. Air Force Office of Scientific Research under Grant FA9550-21-1-0195 (CY, AL).

\begin{appendix}
    \section{Proof of Lemma \ref{lem:clus}}\label{app:A}

We prove \eqref{eq:p=beta} directly for the general adaptive case \eqref{eq:marginal}. We use the shorthand $j=n+1$, and: \begin{equation}\label{eq:En=}
    E_n=E_{\bm{x}_n},\quad E_j=2E_{x_1\cdots x_n0_{j}}-I, \quad p=p(0|\bm{x}_n).
\end{equation}
We closely follow \cite{learn_highT21}, whose Theorem 3.1 establishes \eqref{eq:p=beta} for the case $n=0$. The central technique is called \emph{cluster expansion}; key  technical ingredients are found in \cite{Loschmidt_echo23}, and the method was sketched in the main text. The whole proof in \cite{learn_highT21} is rather long (and uses slightly different notation), but only several places need to be modified to properly insert $E_n$.  Therefore, we will frequently refer to \cite{learn_highT21}, and only explain the steps of their work that require modification (along with introducing requisite notation to explain the relevant modifications). For the reader familiar with \cite{learn_highT21}, our Lemma may seem to be a natural generalization.

To exploit the local nature of $H$ in \eqref{eq:H=Ha}, we need some definitions for clusters.  An example is sketched in Fig.~\ref{fig:clus} to help familiarize the reader with the construction.

\begin{figure}[t]
\centering
\includegraphics[width=.4\textwidth]{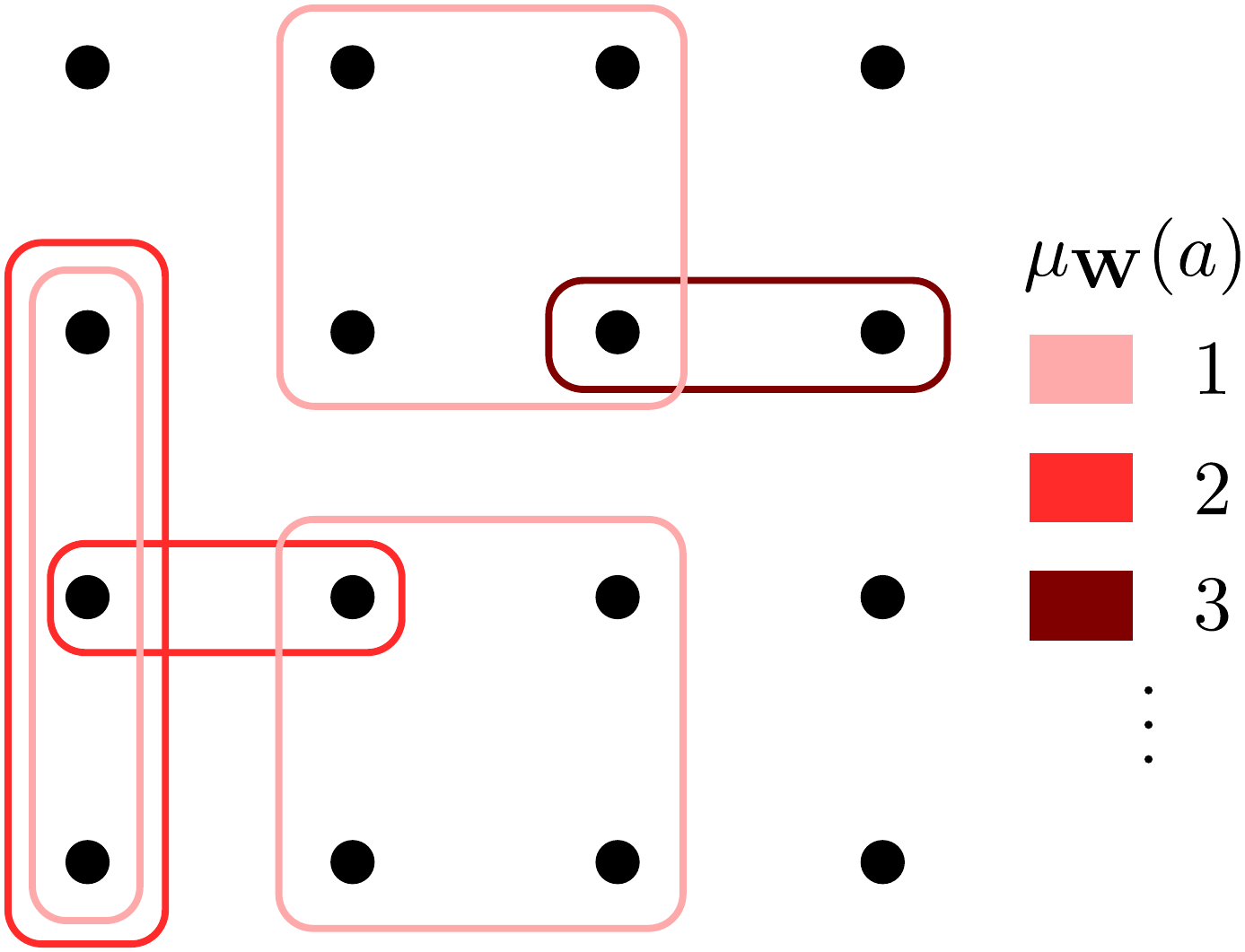}
\caption{\label{fig:clus} Sketch of a cluster $\W$ on a 2d square lattice (black dots). Each Pauli string $a$ contained in $\W$ is shown by a square loop, whose color represents its multiplicity $\mu_\W(a)$. As shown on the very left, two Pauli strings can be supported on an identical set of qubits. One can verify $|\W|=10$ and $\W!=24$ for this example. The cluster $\W$ is disconnected, and is the union $\W=\V_1\cup \V_2$ of two connected clusters, where $\V_1$ ($\V_2$) lies on the lower left (upper right). }
\end{figure}

\begin{defn}
A cluster $\V$ is a multiset of $a$s, which assigns a multiplicity $\mu_\V(a)\in \{0,1,\cdots\}$ to each $a$. We also denote $\V=\{a,\cdots,a,b,\cdots\}$, where $a$ appears $\mu_\V(a)$ times. We say $\V$ contains $a$ (denoted by $a\in \V$) if $\mu_\V(a)\ge 1$. Define the total weight of $\V$ as $|\V|:=\sum_a \mu_\V(a)$, while $\V!:=\prod_a \lr{\mu_\V(a)!}$ (note the difference to $|\V|!$) and $\lambda^\V:=\prod_a \lambda_a^{\mu_{\V}(a)}$. 

Clusters can union with each other. For example, the cluster $\W=\V_1\cup \V_2$ is defined by $\mu_\W(a)=\mu_{\V_1}(a)+\mu_{\V_2}(a), \forall a$. 

A cluster $\W$ is called disconnected, if the $a$s contained in $\W$ can be separated to two sets $A_1,A_2$, such that any pair $a_1\in A_1,a_2\in A_2$ from the two sets do not overlap in their support qubits. Otherwise $\W$ is connected. 
\end{defn}

The proof idea is to first convert the fraction \eqref{eq:p=beta} to a logarithm that is easier to deal with. 
One can explicitly verify (generalizing Proposition 3.2 in \cite{learn_highT21})  \begin{equation}\label{eq:2p-1=log}
    2p-1 = \left.\partial_{\kappa} \log \cZ\right|_{\kappa=0},
\end{equation}
where \begin{equation}\label{eq:cZ=}
    \cZ:= 2^{n-N} \tr{E_n \mathrm{e}^{\kappa E_j} \ee^{-\beta H}},
\end{equation}
and $E_j$ is a Pauli operator on $j$ due to \eqref{eq:En=}.
We proceed to find the $\beta$-expansion for the ``partition function'' $\cZ$ (c.f. (47) in \cite{learn_highT21})
\begin{align}\label{eq:Z=DZ}
    \cZ =& 2^{n-N}\mlr{ \tr{E_n\mathrm{e}^{\kappa E_j}} - \beta \tr{E_n\mathrm{e}^{\kappa E_j}\textstyle{\sum_a} \lambda_a H_a}+ \cdots} \nonumber\\
        =& \cosh\kappa + \sum_{m=1}^\infty \sum_{\V: \abs{\V}=m} \frac{\lambda^\V}{\V!} \mathcal{D}_\V \cZ.
\end{align}
The sum is over all clusters $\V$, where
\begin{equation}\label{eq:DZ=}
    \mathcal{D}_\V \cZ:= 2^{n-N} \frac{(-\beta)^m}{m!} \sum_{\sigma\in \mathsf{S}_m} \tr{E_n \mathrm{e}^{\kappa E_j}H_{a_{\sigma_1}} \cdots H_{a_{\sigma_m}} },
\end{equation}
with $\V=\{a_1,\cdots,a_m\}$, and $\mathsf{S}_m$ being the permutation group of $m$ elements. Intuitively, $\mathcal{D}_\V \cZ$ is a derivative of $\cZ$ over multi-variables $\{\lambda_a\}$ at $\lambda_a=0$.
To get \eqref{eq:Z=DZ}, we have expanded $H$ to local Paulis $H_a$ by \eqref{eq:H=Ha}, used $\tr{E_n}=2^{N-n}$, and organized terms into clusters according to \eqref{eq:DZ=}. The $\V!$ factor appears in \eqref{eq:Z=DZ} because, for example, there is only one $H_a H_a$ term in $H^2$, but that term is counted twice in \eqref{eq:DZ=}.

Taking the logarithm of \eqref{eq:Z=DZ} and expanding in $\beta$, each final term may correspond to multiple clusters $\V_1,\V_2,\cdots$ in \eqref{eq:Z=DZ}, which comes from higher orders in the expansion $\log(1+y)=y-y^2/2+\cdots$ with $y=\sum_{\V}\cdots$. These multiple clusters combine to a single one $\W=\V_1\cup\V_2\cup\cdots$, and we organize the final terms according to the combined cluster: (see (25) in \cite{learn_highT21}) \begin{align}\label{eq:logZ=}
    \log\cZ =\log\cosh\kappa +  \sum_{m=1}^\infty \beta^m \sum_{\W: \abs{\W}=m}\lambda^\W \varGamma_{\W}.
\end{align}
where $\varGamma_{\W}$ do not depend on $\beta$ and $\{\lambda_a\}$. Schematically, \begin{align}\label{eq:gamma=chi}
    \beta^m \varGamma_\W = \sum_{\{\V_\ell\} } \chi(\W, \{\V_\ell\}) \prod_{\ell=1}^L \frac{ \mathcal{D}_{\V_\ell} \cZ}{\cosh\kappa},
\end{align}
where $\{\V_\ell\}$ denotes a multiset of connected clusters $\V_1,\V_2,\cdots,\V_L$ with $L\le m$ that partitions the entire cluster $\W$,\footnote{This partition has several constraints, which we do not elaborate here. See \cite{learn_highT21,Loschmidt_echo23}. }
 and the coefficient $\chi(\W, \{\V_\ell\})$ is determined by geometry alone.
Using \eqref{eq:2p-1=log} to compare to \eqref{eq:p=beta}, we identify 
\begin{equation}\label{eq:gamma=W}
    \gamma_m = \sum_{\W\ni j: \abs{\W}=m} \lambda^\W \lr{\partial_\kappa \varGamma_\W}_0,
\end{equation}
where $\lr{\cdots}_0$ means taking $\kappa=0$.
Here the sum is constrained to clusters $\W$ containing $j$, because if $j\notin \W$, $\mathcal{D}_{\V_\ell}\cZ$ in \eqref{eq:gamma=chi} will depend on $\kappa$ simply by a multiplicative factor $\cosh \kappa$ according to \eqref{eq:DZ=}: \begin{equation}\label{eq:tr=coshk}
    \tr{E_n \mathrm{e}^{\kappa E_j} O} \propto \tr{E_n \mathrm{e}^{\kappa E_j}} \tr{E_n O}\propto \cosh \kappa,
\end{equation}
if $O=H_{a_{\sigma_1}}\cdots$ does not act on $j$. This factor cancels the $\kappa$-dependence in the denominator of \eqref{eq:gamma=chi}, thus taking the $\kappa$-derivative simply yields zero.

Although general expressions for $\varGamma_\W$ are complicated, we know it always \emph{vanishes} if $\W$ is \emph{disconnected}: Suppose the $a$s contained in $\W$ belong to two disconnected sets of qubits $A_1,A_2$, and we set $\lambda_a=0$ for all $a\notin \W$ \footnote{We are free to do so, because the term $\varGamma_\W$ is still included in the logarithm of the restricted partition function.}, so that $H=H_1+H_2$ where $H_1$ ($H_2$) only acts inside $A_1$ ($A_2$). Then similar to \eqref{eq:tr=coshk}, the trace factorizes \begin{align}\label{eq:disconn_factor}
    \tr{E_n \ee^{-\beta H}} &= \Tr\mlr{ \lr{ E_{n,1} \ee^{\kappa E_j} \ee^{-\beta H_1} }\otimes \lr{ E_{n,2}\ee^{-\beta H_2}} } \nonumber\\
    &\propto \tr{E_n \ee^{\kappa E_j} \ee^{-\beta H_1}} \tr{E_n \ee^{-\beta H_2}},
\end{align}
where we have expressed the single-site projectors $E_n=E_{n,1}\otimes E_{n,2}$ with $E_{n,1}$ ($E_{n,2}$) acting inside (outside) $A_1$, and assumed $j\in A_1$ without loss of generality.
As a result, $\log\cZ=\zeta_1 + \zeta_2$, where $\zeta_{1,2}$ is a multivariate polynomial of $\lambda_a$ for $a\in A_{1,2}$ alone. There are no cross terms like $\lambda_{a_1}\lambda_{a_2}$ where $a_1\in A_1,a_2\in A_2$, which proves $\varGamma_\W=0$ because it is a coefficient of such cross terms.

Hence, the sum in \eqref{eq:gamma=W} is further restricted to \emph{connected} clusters of total weight $m$ that contains $j$. There are at most \begin{equation}\label{eq:Cm}
    C_m:=\ee \fkd\mlr{1+\ee(\fkd-1)}^{m-1} \le (\ee \fkd)^{m}
\end{equation} 
of them (Proposition 3.6 in \cite{learn_highT21}), which can be enumerated in time $\order(m\fkd C_m)$ (Section 3.4 in \cite{learn_highT21}). 

Therefore, to prove (\ref{eq:gamma<beta}), we only need to show that for each such cluster $\W$,
\begin{equation}\label{eq:Gamma<}
    \abs{\lr{\partial_\kappa \varGamma_{\W}}_0} \le m\mlr{2\ee (\fkd+1)}^{m}.
\end{equation} 
To proceed, we derive from \eqref{eq:gamma=chi} \begin{widetext}
\begin{align}\label{eq:gamma<betam}
    \beta^m \abs{\lr{\partial_\kappa \varGamma_\W}_0} &= \abs{ \sum_{\{\V_\ell\} } \chi(\W, \{\V_\ell\}) \sum_{\ell'=1}^L \lr{\partial_\kappa\mathcal{D}_{\V_{\ell'}} \cZ}_{0} \prod_{\ell=1, \ell\neq \ell' }^L \lr{\mathcal{D}_{\V_\ell} \cZ}_0 } \nonumber\\
    &\le m \sum_{\{\V_\ell\} } \abs{\chi(\W, \{\V_\ell\})} \beta^{\sum_\ell \abs{\V_\ell}} = m \beta^m \sum_{\{\V_\ell\} } \abs{\chi(\W, \{\V_\ell\})}.
\end{align}
\end{widetext}
Here we have used \begin{equation}\label{eq:DZ<}
    \abs{\lr{\partial_\kappa\mathcal{D}_\V \cZ}_0}, \abs{\lr{\mathcal{D}_\V \cZ}_0}\le \beta^{|\V|},
\end{equation}
together with $L\le m$ and the fact that $\{V_\ell\}$ partitions $\W$: $m=\sum_\ell \abs{\V_\ell}$. We will prove \eqref{eq:DZ<} shortly; for now observe that what we need is a bound on the coefficient $\chi$. This is done precisely in \cite{learn_highT21}, where the major chunk of the proof establishes the right hand side of \eqref{eq:Gamma<} as a bound for \eqref{eq:gamma<betam}. Note that our $\chi$ is the same as what can be read out in eq.(50) of \cite{learn_highT21}, since it is merely a combinatorial factor that comes from expanding the logarithm, and does not depend on how $\mathcal{D}_\V \cZ$ is defined in \eqref{eq:DZ=} (where there is no $E_n\ee^{\kappa E_j}$ insertion in \cite{learn_highT21}). 

Now we show \eqref{eq:DZ<}. Let $m=|\V|$.  \eqref{eq:DZ=} leads to \begin{align}\label{eq:DZ<n}
    \abs{\lr{\mathcal{D}_\V \cZ}_0} &\le 2^{n-N} \beta^m \norm{E_n H_{a_1} \cdots H_{a_m}}_1 \nonumber\\
    &\le 2^{n-N} \beta^m \norm{E_n }_1 \norm{H_{a_1} \cdots H_{a_m}}_\infty = \beta^m.
\end{align}
Here in the first line, we have used that the number of permutations is $|\mathsf{S}_m|=m!$ and $\tr{A}\le \norm{A}_1:=\tr{\sqrt{A^\dagger A}}$. In the second line, we have used H\"{o}lder inequality $\norm{AB}_1 \le \norm{A}_1\norm{B}_\infty$, together with $\norm{H_{a_1} \cdots H_{a_m}}_\infty \le \norm{H_{a_1}}_\infty \cdots \norm{H_{a_m}}_\infty=1$ and $\norm{E_n}_1=\tr{E_n}=2^{N-n}$. As one can verify, $\lr{\partial_\kappa\mathcal{D}_\V \cZ}_0$ is bounded similarly by inserting $E_j$ in \eqref{eq:DZ<n}. As a result, \eqref{eq:DZ<} then establishes \eqref{eq:Gamma<} following \cite{learn_highT21}, from which \eqref{eq:gamma<beta} is obtained.

Finally, to prove Lemma \ref{lem:clus}, we need to show that $\lr{\partial_\kappa\varGamma_\W}_0$ can be computed in time \begin{equation}\label{eq:m2k-1}
  t_{\mathrm{compute}}=   \order\mlr{2^{m(2k-1)} \poly{m} }.
\end{equation} 
\cite{learn_highT21} (Section 3.6) shows that $\varGamma_\W$ is computable in time \begin{equation}\label{eq:8mpolym}
     \order(8^m \poly{m}),
 \end{equation} 
 without the $E_n\ee^{\kappa E_j}$ insertion, via taking multivariate derivative $\mathcal{D}_\W$ on $\log \cZ$ (see \eqref{eq:logZ=}). Here we show that the insertion does not change this scaling much. The only subroutine in \cite{learn_highT21} where $E_n \ee^{\kappa E_j}$ would appear is to calculate \begin{equation}\label{eq:MK=}
    \tr{E_n\ee^{\kappa E_j} \mathcal{M}^{\mathcal{K} } } = \tr{E_n \mathcal{M}^{\mathcal{K} } } + \kappa \tr{E_nE_j \mathcal{M}^{\mathcal{K} } }+\cdots,
\end{equation}
for $\mathcal{K}=1,\cdots,m$, where the matrix $\mathcal{M}:=\sum_{a\in \W} z_a H_a$ with $z_a$ being integers. We have expanded \eqref{eq:MK=} in $\kappa$, so that what is really computed is the $c$-number coefficients. This is more complicated than \cite{learn_highT21} where \eqref{eq:MK=} is a single $c$-number $\tr{\mathcal{M}^{\mathcal{K}}}$. However, this only yields an extra constant factor for the total complexity, because only the first derivative in $\kappa$ is needed in the end, so in each step of the calculation one can throw away $\kappa^2$ and higher orders. 
In other words, we only need to show that computing the first two orders in \eqref{eq:MK=} requires a similar runtime as $\tr{\mathcal{M}^{\mathcal{K}}}$, which is $\order(4^m \poly{m})$ shown in \cite{learn_highT21}. 

We focus on the first term in \eqref{eq:MK=}, since the second follows similarly. Observe that computing that term can be restricted to no more than $1+m(k-1)$ qubits that support $\W$, since $E_n$ is a tensor product. This number of qubits is because the connected cluster $\W$ can be generated by starting from one $k$-local term $a\in \W$, and then gradually appending terms in $\W$ to it while keeping connectivity; each time appending one term adds at most $k-1$ support qubits. In other words, the first term in \eqref{eq:MK=} is restricted to a Hilbert space of dimension $D=\order[2^{m(k-1)}]$. Furthermore, the matrix $\mathcal{M}=\sum_a z_a H_a$ to be powered is sparse: each row/column contains at most $m$ nonzero entries, so multiplying any matrix to it yields complexity $\order(D^2)=\order[4^{m(k-1)}]$. This is the dominant complexity of \eqref{eq:MK=} up to $\poly{m}$, so comparing to $\order(4^m)$ without $E_n \ee^{\kappa E_j}$ insertion, the computation time for $\lr{\partial_\kappa\varGamma_\W}_0$ is given by \eqref{eq:m2k-1} comparing to \eqref{eq:8mpolym} following \cite{learn_highT21}.

\section{Proof of Theorem \ref{thm:nocorr}}\label{app:B}

We sketch the proof since it is very similar to Appendix \ref{app:A}. Define $\cZ'$ in a similar way as \eqref{eq:cZ=} with replacement \begin{equation}
    \kappa E_j\rightarrow \kappa_iE_i+\kappa_j E_j,
\end{equation}
where $E_i,E_j$ are the ones achieving the maximum in \eqref{eq:cor<},
and define $\varGamma'_\W$ as the coefficient of the corresponding expansion in \eqref{eq:logZ=}.
Observe that \begin{align}\label{eq:cor=}
    &\mathrm{Cor}(i,j) = \lr{\partial_{\kappa_i} \partial_{\kappa_j} \log \cZ'}_0 \nonumber\\
    &= \sum_{m=1}^\infty \beta^{m} \sum_{\substack{ \W \text{ connected}:\\ \{i,j\}\subset\W, \abs{\W}=m}}\lambda^\W \lr{\partial_{\kappa_i} \partial_{\kappa_j}\varGamma_{\W}'}_0,
\end{align}
Here the first line comes from explicit calculation similar to \eqref{eq:2p-1=log}, while the second line follows from generalizing the arguments around \eqref{eq:tr=coshk} and \eqref{eq:disconn_factor}. For a connected cluster $\W$ that contains $i$ and $j$, the total weight cannot be too small: $m=\mathrm{\Omega}(\mathsf{d}(i,j))$. As a result, \eqref{eq:cor<} then follows from \eqref{eq:Cm} and a generalization of \eqref{eq:Gamma<}: \begin{equation}
    \abs{\lr{\partial_{\kappa_i} \partial_{\kappa_j} \varGamma'_{\W}}_0} \le m^2\mlr{2\ee (\fkd+1)}^{m}.
\end{equation}
This holds by generalizing \eqref{eq:gamma<betam} to two derivatives, and \eqref{eq:DZ<} to \begin{equation}
    \abs{\lr{\partial_{\kappa_i}\partial_{\kappa_j}\mathcal{D}_\V \cZ'}_0}, \abs{\lr{\partial_{\kappa_i}\mathcal{D}_\V \cZ'}_0}, \abs{\lr{\mathcal{D}_\V \cZ'}_0}\le \beta^{|\V|},
\end{equation}
which comes from, e.g. $\lr{\mathcal{D}_\V \cZ'}_0=\lr{\mathcal{D}_\V \cZ}_0$ and \eqref{eq:DZ<n}.

    
\end{appendix}

\bibliography{biblio}
\end{document}